\documentclass[runningheads]{llncs}

\usepackage{graphicx}
\usepackage{amsmath,amssymb}
\usepackage{textcomp}
\usepackage{comment}
\usepackage{color}
\usepackage[utf8]{inputenc}

\newcommand{\mymatrix}[2]{\left( \begin{array}{#1} #2 \end{array} \right)}
\newcommand{\myvector}[1]{\mymatrix{c}{#1}}
\newcommand{\myvectorr}[1]{\mymatrix{r}{#1}}

\newcommand{\cent}[0]{\mbox{\textcent}}
\newcommand{\dollar}{\$}

\begin{document}

\title{Improved constructions for succinct affine automata}
\titlerunning{Improved constructions for succinct affine automata}

\author{Abuzer Yakary{\i}lmaz\inst{1,2}\orcidID{0000-0002-2372-252X} }
\authorrunning{A. Yakary{\i}lmaz}

\institute{
Center for Quantum Computer Science, University of Latvia, R\={\i}ga, Latvia \and
QWorld Association, Tallinn, Estonia, \url{https://qworld.net} \\
\email{abuzer@lu.lv}
}

\maketitle

\begin{abstract}
Affine finite automata (AfA) can be more succinct than probabilistic and quantum finite automata when recognizing some regular languages with bounded-error. In this paper, we improve previously known constructions given for the succinctness of AfAs in three ways. First, we replace some of fixed error bounds with arbitrarily small error bounds. Second, we present new constructions by using less states than the previous constructions. Third, we show that any language recognized by a nondeterministic finite automaton (NFA) is also recognized by bounded-error AfAs having one more state, and so, AfAs inherit all succinct results by NFAs. As a special case, we also show that any language recognized by a NFA is recognized by AfAs with zero error if the number of accepting path(s) for each member is exactly the  same number.
\end{abstract}
 
\section{Introduction}

Probabilistic finite automaton (PFA) \cite{Rab63,Paz71} is a linear system implementing non-negative transitions by preserving $\ell_1$-norm where a probabilistic state is represented as a non-negative real-valued column vector with entry summation 1. Similarly, quantum finite automaton (QFA) \cite{SY14,AY15} is also a linear system but it implements complex-valued transitions by preserving $ \ell_2 $-norm where a (pure) quantum state is represented as a complex-valued vector with length 1. Implementing both positive and negative valued transitions creates interference and so some transitions may disappear, which brings certain computational advantages to QFAs over PFAs, e.g., bounded-error QFAs can be exponentially more succinct than bounded-error PFAs \cite{af98}, or nondeterministic QFAs are more powerful than nondeterministic finite automaton (NFA) \cite{YS10A}. 

One may ask whether it is possible to use interference\footnote{We refer the reader to \cite{Hir18} for certain discussions about interference with historical remarks.} classically. The idea of using negative transition values for classical systems dates back to sixties. Turakainen \cite{Tur69} defined generalized automaton (GA) as a linear system implementing real-valued transitions without any restrictions. The language recognition by GAs are defined based on cutpoints, and bounded-error language recognition has never been considered.

After reading the whole input, the final state of a GA is represented as a column vector with real-valued entries. To calculate the accepting value, this vector is multiplied with a pre-defined real-valued row vector (with the same dimension). In other words, each state contributes to the accepting value by a real-valued weight:
\[
    f_G(x) = w \cdot v_f = ( w_1~~~w_2~~~\cdots ~~~ w_n) \cdot \myvector{\alpha_1 \\ \alpha_2 \\ \vdots \\ \alpha_n}
    = \sum_{i=1}^n w_i \cdot \alpha_i ,
\]
where $ G $ is the GA, $ x $ is the input, $ f_G(x) $ is the accepting value of $ G $ on $ x $, $ w $ is the pre-defined weights, and $ v_f $ is the final state. Remark that for PFAs, $ w $ contains only 0s and as, where 1s are corresponding to the accepting states. 

In the above, $ f_G(x) $ is in $ \mathbb{R} $. On the other hand, an accepting probability is in $ [0,1] $. One way to observe each state with some probabilities when in $ v_f $ (similar to PFAs and QFAs) is making a normalization with respect to $ \ell_1 $-norm. For QFAs, some measurement operators are applied to the quantum state, and then, the different outcomes are observed with some probabilities. For classical case, we define an operator called weighting \cite{DCY16} produces the outcomes with probabilities based on their normalized values in $ \ell_1 $. Here we should remark that, on contrary to measurement operators, weighting is a non-linear operator. 

Affine finite automaton (AfA) is a new and quantum-like generalization of PFA which evolves linearly followed by a non-linear weighting operator \cite{DCY16}. An affine state can have arbitrary real numbers but the summation of them must be 1 similar to the probabilistic state.

The computational power of AfAs and their generalizations have been examined and compared with their probabilistic and quantum counterparts in a series of papers \cite{VY16,BMY17,HMY17,NKPVY17,IKPY18,VY18,HMY19,HMY21,IKPY21,KY21A}. In the unbounded-error and bounded-error language recognition modes, AfAs are more powerful than PFAs and QFAs, where the latter models recognize all and only stochastic and regular languages, respectively \cite{Rab63,Paz71,KW97,YS11A,LQZLWM12}. In the nondeterministic language recognition mode, PFAs can be seen as nondeterministic finite automata (NFAs), and so they recognize all and only regular languages. On the other hand, the class of languages recognized by nondeterministic QFAs is a superset of regular languages known as co-exclusive stochastic languages \cite{Paz71,YS10A}. On contrary to bounded-error and unbounded-error cases, nondeterministic AfAs are equivalent to nondeterministic QFAs \cite{DCY16}. 

Regarding the state complexity, bounded-error PFAs can be exponentially more succinct than deterministic finite automata (DFAs) \cite{Rab63,Paz71,Amb96}; and, bounded-error QFAs can be exponentially more succinctness than bounded-error PFAs \cite{af98,AY15}. On the other hand, bounded-error AfAs can be more succinct than bounded-error PFAs and QFAs and the gap can be super-exponential \cite{BMY17,VY18,IKPY21}. A similar gap between DFAs and bounded-error PFAs and between bounded-error PFAs and bounded-error QFAs can only be obtained on promise problems \cite{AmbY12,GefY15}.

In this paper, we improve previously known constructions given for the succinctness of AfAs in three ways. In the next section, we give the definitions and notations used throughout the paper. In Section~\ref{sec:simulations}, we quickly review the simulations by AfAs that we refer in the rest of paper. In Section~\ref{sec:three-state-AfAs}, we give our constructions for three-state AfAs. In Section~\ref{sec:NFA}, we present our results on NFAs.

 \section{Preliminaries}
 \label{sec:pre} 
 
We assume the reader familiar with the basic of automata theory. Throughout the paper, we denote the input alphabet as $ \Sigma $ not including $ \cent $ (the left end-marker) and $ \dollar $ (the right end-marker), and we denote $ \Sigma \cup \{\cent,\dollar\} $ as $ \Tilde{\Sigma} $. For a given input $ x \in \Sigma $, $ \tilde{x} $ denotes $ \cent x \dollar $. For a given string $ x $, $ |x| $ is its length; for a numeric value $ \alpha $, $ |\alpha| $ is the absolute value of $\alpha$; and, for a vector $ v $, $ |v| $ is $ \ell_1 $-norm of $ v $. For a non-empty string $ x $, $ x[i] $ denotes its $i$-th symbol, where $1 \leq i \leq |x|$. For a given matrix $ A $, $ A[i,j] $ is its entry at the $ i $-th row and $ j $-th column; and, for a given vector $ v $, $ v[i] $ is its $i$-th entry and $ \zeta(v) $ is the summation of all entries.
For an automaton $ M $ and input string $x$, $ f_M(x) $ is the accepting probability of $ M $ on $x$.

An affine state is a real-valued column vector with entry summation 1. An affine operator is a real-valued square matrix where each column is an affine state. If we use only non-negative values, then an affine state is a probabilistic state (also called stochastic vector) and an affine operator is a probabilistic operator (also called stochastic matrix).

An $n$-state affine finite automaton (AfA)\footnote{We use lowercase ``f'' to emphasis non-linear behaviours of the automaton.} $M$ is a 5-tuple
\[
    M = (S,\Sigma,\{ A_\sigma \mid \sigma \in \Tilde{\Sigma} \},s_I,S_a),
\]
where 
\begin{itemize}
    \item $ S = \{s_1,\ldots,s_n\} $ is the set of states,
    \item $ A_\sigma $ is the affine operator when reading symbol $ \sigma \in \Tilde{\Sigma} $,
    \item $ s_I \in S $ is the initial state, and
    \item $ S_a \subseteq S $ is the set of accepting state(s).
\end{itemize}

Let $ x \in \Sigma^* $ be the input with length  $m$. The automaton $ M $ starts in affine state $ v_0 $, which is the elementary basis $ e_I $ in $ \mathbb{R}^n $. If $ x $ is the empty string, then the final state is $ v_f = A_\dollar A_{\cent} v_0 $. Otherwise, the final state is calculated as
\[
    v_f = A_\dollar A_{x[m]}A_{x[m-1]} \cdots A_{x[1]} A_{\cent} v_0.
\]
The input is accept with probability
\[
    f_M(x) = \frac{\sum_{s_i \in S_a} |v_f[i]| }{|v_f|}.
\]

If we use only non-negative transition values, then we obtain a probabilistic finite automaton (PFA). If we use only 0s and 1s, then we obtain a deterministic finite automaton (DFA).

A language $ L \subseteq \Sigma^* $ is said to be recognized by an automaton $ M $ with error bound $ \epsilon < \frac{1}{2} $ if (i) for each $ x \in L $, $ f_M(x) \geq 1 - \epsilon $ and (ii) for each $ x \notin L $, $ f_M(x) \leq \epsilon $.

A language $ L \subseteq \Sigma^* $ is said to be recognized by an automaton $ M $ with positive one-sided error bound $ \epsilon < 1 $ if (i) for each $ x \in L $, $ f_M(x) \geq 1 - \epsilon $ and (ii) for each $ x \notin L $, $ f_M(x) =0 $.

A language $ L \subseteq \Sigma^* $ is said to be recognized by an automaton $ M $ with negative one-sided error bound $ \epsilon < 1 $ if (i) for each $ x \in L $, $ f_M(x) = 1 $ and (ii) for each $ x \notin L $, $ f_M(x) < \epsilon $.

When $ \epsilon = 0 $, then it is called zero error.

We call an automaton ((positive/negative)  one-sided) bounded-error, if it recognizes its language in the specified error mode with some error bounds.

\section{Simulations}
\label{sec:simulations}

We review the basic simulations by AfAs. We start with a generic case.

\subsection{A sequence of matrix-vector multiplication}
\label{sec:linear}

Let $ v_0 $ be a real-valued $ n $-dimensional column vector and let $ A_1,\ldots,A_k $ be some $(n \times n)$-dimensional real-valued linear operators. We define affine vector $ v_0' $ and affine operator $ A_i' $ ($ 1 \leq i \leq k $) as
\[
    v_0' = \myvector{ v_0[1] \\ \vdots \\ v_0[n] \\ \hline 1 - \zeta(v_0) } 
    ~~ \mbox { and } ~~
    A_i' = \mymatrix{ccc|c}{ 
    c_1[0] & \cdots & c_n[0] & 0 \\ 
    \vdots & \ddots & \vdots & \vdots \\
    c_1[n] & \cdots & c_n[n] & 0 \\ 
    \hline 1 - \zeta(c_1) & ~\cdots~ & 1 - \zeta(c_n) & 1 },
\]
where $ c_j $ is the $ j $-th column of $ A_i $, where $ 1 \leq j \leq n $. 
Then, for given $v_f = A_k A_{k-1} \cdots A_1 v_0$, 
we can have
\[
    v_f' =  A_k' A'_{k-1} \cdots A'_1 v'_0 =  \myvector{ v_f[1] \\ \vdots \\ v_f[n] \\ \hline 1 - \zeta(v_f) }.
\]

\subsection{Trivial case for PFAs}
\label{sec:trivial-PFA}

It is trivial that any $n$-state PFA is an  $n$-state AfA. So, PFAs and DFAs cannot be more succinct than bounded-error AfAs.

\subsection{Rational exclusive stochastic languages}
\label{sec:exclusive}

Let $ L $ be a language defined by an $n$-state rational-valued PFA $ P $ and a cutpoint $ \lambda \in [0,1] $ as given below:
\[
    L = \{ w \mid f_P(w) \neq \lambda \}.
\]
Based on the simulation given in Section~\ref{sec:linear}, it was shown \cite{VY18} that $ L $ is recognized by an $ (n+1) $-state integer-valued AfA as follows:
\begin{itemize}
    \item each $ x \in L $ is accepted by the AfA with probability no less than $ \frac{1}{3} $, and,
    \item each $ x \notin L $ is accepted by the AfA with zero probability.
\end{itemize}

\subsection{Exact simulation of QFAs}
\label{sec:exact-QFA}

Any given $ n $-state QFA can be simulated exactly by a $ (n^2+1) $-state AfA \cite{VY18}. The computation of a QFA is linear. By tensoring the computation with itself, the probabilities can be directly accessed on the state vectors. Each complex number can be represented by two real numbers, but the tensoring vectors have some redundancy and so $ n^2 $-dimensional real-valued vectors can be obtained from $ n $-dimentional quantum state. The rest of the proof is due to Section~\ref{sec:linear}. If the QFA is real-valued, we still do not know any better bound.

Potentially bounded-error QFAs can be quadratically more succinct than bounded-error AfAs, but it is open whether QFAs can be more succinct than AfAs or whether any $n$-state QFA can be simulated by a $ \Theta(n) $-state or $ o(n) $-state AfAs.

\section{Three-state AfAs}
\label{sec:three-state-AfAs}

In this section, we give improved constructions of 3-state AfAs for some unary languages.

We start with the well-known counting problem: $ \mathtt{COUNT_m} = \{a^m\} $ for some $ m \geq 0 $. It was shown \cite{VY18} that the language $ \mathtt{COUNT_m} $ is recognized by a 2-state AfA with (negative) one-sided error bound $ \frac{1}{3} $. We decrease the error bound arbitrarily by using one more state.

\begin{theorem}
    The language $ \mathtt{COUNT_m} $ is recognized by a 3-state AfA with (negative) one-sided error bound $ \frac{1}{2t+1} $ for some $ t \in \mathbb{Z}^+ $.
\end{theorem}
\begin{proof}
    The affine states are $ s_1 $, $s_2$, and $s_3$, where $ s_1 $ is the initial and only accepting state. The initial affine state is $ v_0= (1~~0~~0)^T $. After reading $ \cent $, the affine state is set to
    \[
        v_1 = \myvectorr{ 1 \\ m \\ -m } = \mymatrix{rrr}{1 & ~~0 & ~~0 \\ m & 1 & 0 \\ -m & 0 & 1} \myvector{1 \\ 0 \\ 0}.
    \]
    For each symbol $ a $, the value of $ e_2 $ (resp., $e_2$) is decreased (resp., increased) by 1 by using the following operator
    \[
        A_a = \mymatrix{rrr}{1 & ~~0 & ~~0 \\ -1 & 1 & 0 \\ 1 & 0 & 1}, \mbox{ i.e., }
        \myvectorr{ 1 \\ t-1 \\ 1-t  } = 
        \mymatrix{rrr}{1 & ~~0 & ~~0 \\ -1 & 1 & 0 \\ 1 & 0 & 1} 
        \myvectorr{ 1 \\ t \\ -t  }.
    \]
    Let $ l $ be the length of the input. Then, the affine state before reading $ \dollar $ is
    \[
        v_{l+1} = \myvector{1 \\ m-l \\ l-m  }.
    \]
    After reading $ \dollar $ symbol, the values of $ e_2 $ and $ e_3 $ are multiplied by $ t $:
    \[
        v_f = \myvector{ 1 \\ t(m-l) \\ t (l-m) } =
        \mymatrix{rrr}{1 & ~-t & ~-t \\ 0 & t & 0 \\ 0 & 0 & t} \myvector{1 \\ m-l \\ l-m  }.
    \]
    If $ l=m $, then $ v_f = v_0 $ and the input is accepted with probability 1. Otherwise, $ |m-l| \geq 1 $, and so, the accepting probability is at most $ \dfrac{1}{2t+1} $.
    \qed
\end{proof}

It is clear that the number of states required by bounded-error PFAs and QFAs recognizing $ \mathtt{COUNT_m} $ increases when $ m $ increases. Similar to AfAs, two-way QFAs can recognize $ \mathtt{COUNT_m} $ with a few states in polynomial expected time in $m$ \cite{YS10B}.

We will continue with language $ \mathtt{MOD}_p = \{ a^{j\cdot p} \mid j \in \mathbb{N} \} $ for some prime number $ p $. This language is recognized by QFAs with $ O(\log p) $ states and bounded-error PFAs require at least $ p $ states \cite{af98}. Previously, the bound for AfA was given by using the simulation given in Section~\ref{sec:exact-QFA} \cite{VY18}. Here, we show that we can indeed use only 3 states.

\begin{theorem}
    The language $ \mathtt{MOD}_p $ is recognized by a 3-state AfA with (negative) one-sided error bound $ \dfrac{ \cot( \pi/p ) }{ t} $ for some $ t>1 $.
\end{theorem}
\begin{proof}
    We use the single qubit algorithm given for this problem \cite{af98}. By help of one more state, we will trace the computation by affine states, which also helps us to decrease the accepting probability arbitrarily for the non-members.
    
    Let $ \{s_1,s_2,s_3\} $ be our states and let $ \theta = \frac{2\pi}{p} $ be our rotation angle. We start in affine state $ v_0 = (1~~0~~0)^T $ and we apply the identity operator when reading symbol $ \cent $. Then, for each symbol $a$, we apply the following operator that implements the counter-clockwise rotation with angle $ \theta $ on the unit circle by using $ s1 $ and $s_2$ :
    \[
        A_a = \mymatrix{cc|c}{ \cos \theta & -\sin \theta & 0 \\ \sin \theta & \cos \theta & 0 \\ \hline \alpha_1 & \alpha_2 & 1    },
    \]
    where $ \alpha_1 = 1 - \cos\theta - \sin\theta $ and $ \alpha_2 = 1 + \sin\theta - \cos\theta $.
    
    Let $ l $ be the length of input. Before reading $ dollar $ symbol, the affine state is 
    \[
        v_{l+1} = \myvector{ \cos ( l \theta ) \\ \sin ( l \theta ) \\ 1 - \cos ( l \theta ) -  \sin ( l \theta )}
    \]
    After reading $ \dollar $ symbol, the final affine state is set to 
    \[
        v_{f} = \myvector{ \cos ( l \theta ) \\ t \sin ( l \theta ) \\ 1 - \cos ( l \theta ) -  t \sin ( l \theta )}.
    \]
    For members, $ v_f = v_0 $ and so the input is accepted with probability 1. For non-members, the ratio 
    \[
        \frac{|\cos ( l \theta )|}{|\sin ( l \theta )|} \geq 
        \frac{ \cos(\frac{\pi}{p}) }{ \sin(\frac{\pi}{p})  } = \cot (\frac{\pi}{p}),
    \]
    where the bound is obtained when
    \[ 
        l \equiv \dfrac{p-1}{2} \mod p ~~\mbox{ or }~~ 
        l \equiv \dfrac{p+1}{2} \mod p,
    \]
    i.e., the rotating vector is at its closest points to the $ x $-axis. Thus, the accepting probability is less than $ \dfrac{ \cot( \pi/p ) }{ t} $. \qed
\end{proof}

We close this section with a promise problem given in \cite{AmbY12}: For any $ k \in \mathbb{Z}^+$,
\[
    \mathtt{MOD2^k} = (\mathtt{0MOD2^k},\mathtt{1MOD2^k}),
\]
where $ \mathtt{0MOD2^k} = \{ a^{j\cdot 2 ^k} ~\mid~ j \equiv 0 \mod 2 \} $ and  $ \mathtt{1MOD2^k} = \{ a^{j\cdot 2 ^k} ~\mid~ j \equiv 1 \mod 2 \} $. This promise problem is solved by 2-state QFAs with zero error \cite{AmbY12}, and different types of classical automata require $ 2^{k+1} $ states to solve it \cite{GefY15}.

By using simulation in Section~\ref{sec:exact-QFA}, it was given in \cite{VY18} that 5-state AfA can solve this problem with zero error. We believe that 2-state AfAs cannot solve this problem with zero error. Here we give a 3-state AfAs with zero error.

\begin{theorem}
    For a given $ k \in \mathbb{Z}^+$, the promise problem $ \mathtt{0MOD2^k} $ is solved by an AfA with zero error.
\end{theorem}
\begin{proof}
    The 2-state QFA algorithm uses a rotation with angle $ \dfrac{\pi}{2^{k+1}} $ on the unit circle \cite{AmbY12}:
    \[
    \myvector{1 \\ 0} \xrightarrow{2^{k}~symbols} \myvector{0 \\ 1} \xrightarrow{2^{k}~symbols}
    \myvector{-1 \\ 0}  \xrightarrow{2^{k}~symbols}
    \myvector{0 \\ -1 }  \xrightarrow{2^{k}~symbols}
    \myvector{1 \\ 0}.
    \]
    Thus, the outcomes alternates between the states ``0'' or ``1'' for each block of $ 2^k $ symbols. Remark that the measurement results for $ \myvector{1 \\ 0 } $ and $ \myvector{-1 \\ 0} $ are the same. But, if we use the simulation given in Section~\ref{sec:linear}, the affine states for members of $\mathtt{0MOD2^k}$ will be 
    \[
    \myvector{1 \\ 0 \\ 0} ~\mbox{ and }~
    \myvectorr{-1 \\ 0 \\ 2},
    \]
    which are different each other.
    
    Instead of the rotating with angle $ \dfrac{\pi}{2^{k+1}} $, we use a rotation with angle $ \dfrac{\pi}{2^{k}} $. Then, we will have the following cycle:
    \[
    \myvector{1 \\ 0} \xrightarrow{2^{k}~symbols} \myvector{-1 \\ 0} \xrightarrow{2^{k}~symbols}
    \myvector{1 \\ 0}   \xrightarrow{2^{k}~symbols}
    \myvector{-1 \\ 0}.
    \]
    Quantumly, we visit two states having the same statistics (i.e., they are identical and no measurement can separate them.). But, they are different vectors. Now, we use the simulation given in Section~\ref{sec:linear}, and so we have the following affine states before reading $\dollar$ symbol for the members of $ \mathtt{0MOD2^k} $ and for the members of $ \mathtt{1MOD2^k} $
    \[
    \myvector{1 \\ 0 \\ 0} ~\mbox{ and }~
    \myvectorr{-1 \\ 0 \\ 2},
    \]
    respectively. After reading $\dollar$ symbol, we half the value of $ s_3 $ and add the other half to the value of $ s_1 $ (such trick was used before in \cite{NKPVY17}). Then, these two affine states becomes
    \[
        \myvector{1 \\ 0 \\ 0} ~\mbox{ and }~
    \myvectorr{0 \\ 0 \\ 1},
    \]
    respectively. We make $ s_1 $ the only the accepting state, and so two different cases can be separated with zero error.
\qed\end{proof}

\section{Simulating NFAs}
\label{sec:NFA}

When the cutpoint is picked as 0, then the given PFA in Section~\ref{sec:exclusive} turns out to be a NFA, where each non-zero transition corresponds to a nondeterministic choice (transition). Thus, any succinctness result for NFAs can be obtained for AfAs having one more state with one-sided error bound $\frac{2}{3}$. (To obtain a better error bound, we can tensor a few copies of the same automaton, which increases the number of states polynomially.)

Here, we present a pedagogically easier construction and more importantly with arbitrarily small error bounds without increasing the previous state bound. We also show that if the number of accepting path(s) are the same (e.g., one) for each member, then the error is zero.

An NFA does not use end-markers but use $\varepsilon$-transition(s). By using the left end-marker, all $\varepsilon$-transition(s) without reading any symbol at the beginning of computation can be replaced with the transitions defined for the left end-marker. All other $\varepsilon$-transition(s) can also be removed by defining new transitions (without using any extra states). When using the right end-marker, NFAs may save at most one state, since any NFA using the right end-marker can be simulated by a NFA without the right end-marker by using one extra state: each transition going to an accepting state when reading the right end-marker goes to this new state, which will be the single accepting state. 

We represent the computation of an $n$-state NFA, say $N$, on a given input $ x \in \Sigma^* $ linearly, where $ |x| = l $ and $ n>1 $. We assume that $ N $ does not have any $\varepsilon$-transitions and it uses the left end-marker. We use integer-valued vectors to represent the states of $ N $ and zero-one matrices to represent the transitions of $ N $.

We assume that the set of states of $ N $ is $ S = \{ s_1,\ldots,s_n \} $ and $ s_1 $ is the initial state. Let $ S_a \subseteq S $ be the set of accepting state(s). The vector $ v_0 = (1~~0~~\cdots~~0)^T $ represents the initial ``nondeterministic'' state. For each symbol $ \sigma \in \Sigma \cup \{ \cent \} $, we define the ``nondeterministic'' operator $ A_\sigma $ where $ A_\sigma[j,i] $ is 1 if there is a transition from $ s_i $ to $ s_j $ when reading symbol $ \sigma $, and it is 0, otherwise. Thus, the final nondeterministic state of $ N $ on $ x $ can be calculated as
\[
    v_{l+1} = A_{x[l]} \cdots A_{x[1]} A_{\cent} v_0.
\]
Here $ v_{l+1} $ contains non-negative integers. A nice property of this presentation is that the value for $ s_i $ represents the number of nondeterministic path(s) ending in $ s_i $ at the end. Remark that some paths may be terminated before, which will not be counted on $ v_{l+1} $.

By using the construction in Section~\ref{sec:linear}, we can design an $ (n+1) $-state AfA $ M $ such that its affine state before reading the right end-marker is
\[
    v'_{l+1} = \myvector{ v_l[1] \\ \vdots \\ v_l[n] \\ \hline 1 - \zeta(v_l) }.
\]
Let $ \alpha = \sum_{s_i \in S_a} v_l[i] $, i.e., the summation of all entries corresponding to the accepting state(s) of $ N $. The AfA $ M $ maps $ v'_{l+1} $ to
\[
    v'_f = \myvector{ t\alpha \\ -t \alpha \\ 1 \\ 0 \\ \vdots \\ 0 }
\]
after reading the right end-marker for some $ t \in \mathbb{Z}^+ $. The accepting states of $ M $ are $ \{s_1,s_2\} $. If $ x \in L $, then $ \alpha $ is a positive integer and so $x$ is accepted with probability no less than $ \dfrac{2t}{2t+1} $, which means that the error can be at most $ \dfrac{1}{2t+1} $. If $ x \notin L $, then $ \alpha = 0 $ and so it is accepted with probability 0.

\begin{theorem}
    Let $ L $ be a language recognized by an $ n $-state NFA, where $ n >1 $. Then, $ L $ is also recognized by an $(n+1)$-state AfA with arbitrarily small (positive) one-sided error bound.
\end{theorem}

Suppose that the NFA has a single accepting path for each member. Then, $ \alpha $ is 1 for the members and it is 0 for the non-members. Thus, we can design an zero-error AfA by setting the final affine state as
\[
    v'_f = \myvector{\alpha \\ 0 \\ 1-\alpha \\ 0 \\ \vdots \\ 0}
\]
after reading the right end-marker. It is easy to see that the accepting probability is 1 (resp., 0) for each member (resp., non-member).

\begin{theorem}
    \label{thm:NFA-single-path}
    Let $ L $ be a language recognized by an $ n $-state NFA such that each member is accepted on exactly one nondeterministic path, where $ n > 1 $. Then, $ L $ is also recognized by an $(n+1)$-state AfA with zero error.
\end{theorem}

\begin{corollary}
    Let $ L $ be a language recognized by an $ n $-state NFA such that each member is accepted on exactly $k>1$ nondeterministic paths, where $ n > 1 $. Then, $ L $ is also recognized by an $(n+1)$-state AfA with zero error.
\end{corollary}
\begin{proof}
    The only modification in the above proof is on the final state as
    \[
    v'_f = \myvector{ \dfrac{\alpha}{k} \\ 0 \\ 1-\dfrac{\alpha}{k} \\ 0 \\ \vdots \\ 0}
    \]
    since $ \dfrac{\alpha}{k} $ is 1 for the members, and 0, otherwise.\qed
\end{proof}

Remark that zero-error PFAs and QFAs cannot be more succinct than DFAs \cite{Kla00}. Thus, zero-error AfAs can be exponentially more succinct than zero-error PFAs and QFAs due to the following witness languages.

The language $ \mathtt{MODXOR_k} $ \cite{IKPY21} is formed by the strings
\[
	\{0,1\}^{t} x_1\{0,1\}^{2k-1} x_2\{0,1\}^{2k-1} \cdots x_m \{0,1\}^{2k-1},
\]
where $ t<2k,m > 0 $, each $ x_i \in\{0,1\} $ for $ 1 \leq i \leq m $, and $ \bigoplus_{i=1}^m x_i =1  $. It was shown \cite{IKPY21} that $ \mathtt{MODXOR_k} $ for $ k>0 $ is recognized by a $ (2k+1) $-state AfA with zero error. Due to Theorem~\ref{thm:NFA-single-path}, the same results can also be obtained by designing a $ 2k $-state NFA, which accepts each member on a single path. 

Compared to $ \mathtt{MODXOR_k} $, the language $ \mathtt{END_k} = \{ \{0,1\}^*1\{0,1\}^{n-1} \} $ is much simpler, and we know that it is recognized by an $ n $-state NFA, which accepts each member on a single path, and any DFA (and so any zero-error QFA) requires at least $ 2^n $ states. 

\begin{corollary}
    The language $ \mathtt{END_k} $ is recognized by an $(n+1)$-state AfA with zero error.
\end{corollary}

\section*{Acknowledgements}
Yakary{\i}lmaz was partially supported by the ERDF project Nr. 1.1.1.5/19/A/005 ``Quantum computers with constant memory''. 

\bibliographystyle{splncs04}
\bibliography{tcs}

\begin{thebibliography}{10}
\providecommand{\url}[1]{\texttt{#1}}
\providecommand{\urlprefix}{URL }
\providecommand{\doi}[1]{https://doi.org/#1}

\bibitem{af98}
Ambainis, A., Freivalds, R.: 1-way quantum finite automata: strengths,
  weaknesses and generalizations. In: FOCS'98. pp. 332--341. IEEE (1998)

\bibitem{Amb96}
Ambainis, A.: The complexity of probabilistic versus deterministic finite
  automata. In: {ISAAC}'96. LNCS, vol.~1178, pp. 233--238. Springer (1996)

\bibitem{AmbY12}
Ambainis, A., Yakary{\i}lmaz, A.: Superiority of exact quantum automata for
  promise problems. Information Processing Letters  \textbf{112}(7),  289--291
  (2012)

\bibitem{AY15}
Ambainis, A., Yakary{\i}lmaz, A.: Automata: From mathematics to applications.
  Tech. Rep. 1507.01988, arXiv (2015), to appear in Automata and Quantum
  Computing edited by Jean-{\`E}ric Pin

\bibitem{BMY17}
Belovs, A., Montoya, J.A., Yakary{\i}lmaz, A.: On a conjecture by {Christian}
  {Choffrut}. Int. J. Found. Comput. Sci.  \textbf{28}(5),  483--502 (2017)

\bibitem{DCY16}
D\'{i}az-Caro, A., Yakary{\i}lmaz, A.: Affine computation and affine automaton.
  In: Computer Science --- Theory and Applications. LNCS, vol.~9691, pp. 1--15.
  Springer (2016), arXiv:1602.04732

\bibitem{GefY15}
Geffert, V., Yakary{\i}lmaz, A.: Classical automata on promise problems.
  Discrete Mathematics \& Theoretical Computer Science  \textbf{17}(2),
  157--180 (2015)

\bibitem{Hir18}
Hirvensalo, M.: Interference as a computational resource: a tutorial. Natural
  Computing  \textbf{17}(1),  201--219 (2018)

\bibitem{HMY17}
Hirvensalo, M., Moutot, E., Yakary{\i}lmaz, A.: On the computational power of
  affine automata. In: Language and Automata Theory and Applications. LNCS,
  vol. 10168, pp. 405--417. Springer (2017)

\bibitem{HMY19}
Hirvensalo, M., Moutot, E., Yakary{\i}lmaz, A.: On the computational power of
  affine automata. In: Unconventional Computation and Natural Computation.
  LNCS, vol. 11493, pp. 108--121 (2019)

\bibitem{HMY21}
Hirvensalo, M., Moutot, E., Yakary{\i}lmaz, A.: Computational limitations of
  affine automata and generalized affine automata. Natural Computing  (2021),
  https://doi.org/10.1007/s11047-020-09815-1

\bibitem{IKPY18}
Ibrahimov, R., Khadiev, K., Pr\={u}sis, K., Yakary{\i}lmaz, A.: Error-free
  affine, unitary, and probabilistic {OBDDs}. In: Descriptional Complexity of
  Formal Systems. LNCS, vol. 10952, pp. 175--187. Springer (2018),
  arXiv:1703.07184

\bibitem{IKPY21}
Ibrahimov, R., Khadiev, K., Pr\={u}sis, K., Yakary{\i}lmaz, A.: Error-free
  affine, unitary, and probabilistic {OBDDs}. International Journal of
  Foundations of Computer Science  (2021),
  https://doi.org/10.1142/S0129054121500246

\bibitem{KY21A}
Khadieva, A., Yakary\={\i}lmaz, A.: Affine automata verifiers. Tech. Rep.
  2104.11192, arXiv (2021)

\bibitem{Kla00}
Klauck, H.: On quantum and probabilistic communication: Las vegas and one-way
  protocols. In: STOC'00: Proceedings of the thirty-second annual ACM symposium
  on Theory of computing. pp. 644--651 (2000)

\bibitem{KW97}
Kondacs, A., Watrous, J.: On the power of quantum finite state automata. In:
  FOCS'97. pp. 66--75 (1997)

\bibitem{LQZLWM12}
Li, L., Qiu, D., Zou, X., Li, L., Wu, L., Mateus, P.: Characterizations of
  one-way general quantum finite automata. Theoretical Computer Science
  \textbf{419},  73--91 (2012)

\bibitem{NKPVY17}
Nakanishi, M., Khadiev, K., Pr\={u}sis, K., Vihrovs, J., Yakary{\i}lmaz, A.:
  Exact affine counter automata. In: 15th International Conference on Automata
  and Formal Languages. {EPTCS}, vol.~252, pp. 205--218 (2017),
  arXiv:1703.04281

\bibitem{Paz71}
Paz, A.: Introduction to Probabilistic Automata. Academic Press, New York
  (1971)

\bibitem{Rab63}
Rabin, M.O.: Probabilistic automata. Information and Control  \textbf{6},
  230--243 (1963)

\bibitem{SY14}
Say, A.C., Yakary{\i}lmaz, A.: Quantum finite automata: A modern introduction.
  In: Computing with New Resources, pp. 208--222. Springer (2014)

\bibitem{Tur69}
Turakainen, P.: Generalized automata and stochastic languages. Proceedings of
  the American Mathematical Society  \textbf{21},  303--309 (1969)

\bibitem{VY16}
Villagra, M., Yakary{\i}lmaz, A.: Language recognition power and succintness of
  affine automata. In: Unconventional Computation and Natural Computation.
  LNCS, vol.~9726, pp. 116--129. Springer (2016)

\bibitem{VY18}
Villagra, M., Yakary{\i}lmaz, A.: Language recognition power and succinctness
  of affine automata. Natural Computing  \textbf{17}(2),  283--293 (2018)

\bibitem{YS10A}
Yakary{\i}lmaz, A., Say, A.C.C.: Languages recognized by nondeterministic
  quantum finite automata. Quantum Information {\&} Computation
  \textbf{10}(9\&10),  747--770 (2010)

\bibitem{YS10B}
Yakary{\i}lmaz, A., Say, A.C.C.: Succinctness of two-way probabilistic and
  quantum finite automata. Discrete Mathematics \& Theoretical Computer Science
   \textbf{12}(2),  19--40 (2010)

\bibitem{YS11A}
Yakary{\i}lmaz, A., Say, A.C.C.: Unbounded-error quantum computation with small
  space bounds. Information and Computation  \textbf{279}(6),  873--892 (2011)

\end{thebibliography}

\end{document}